\documentclass[11pt,unknownkeysallowed]{article}
\usepackage{ifthen,calc}
\usepackage{amsmath,amsthm,amssymb,amsfonts,bm,commath}
\usepackage[margin=1.2in]{geometry}

\usepackage{caption2}

\usepackage{floatrow}
\newfloatcommand{capbtabbox}{table}[][\FBwidth]

\usepackage{comment}
\usepackage{enumitem}
\setlist{nosep}
\usepackage{setspace}   
\allowdisplaybreaks

\usepackage{xcolor,graphicx}
\usepackage{pgf,tikz}
\usetikzlibrary{calc,topaths,arrows}
\tikzset{>=latex}
\usepackage{placeins}   
\usepackage{pgfplots}   
\usepackage[overload]{empheq}

\usepackage{natbib}

\newtheorem{thm}{Theorem}
\newtheorem{prop}{Proposition}
\newtheorem{lem}{Lemma}

\newtheorem{defn}{Definition}
\newtheorem{example}{Example}
\newtheorem{remark}{Remark}

\usepackage{hyperref}
\hypersetup{colorlinks=true,
            citecolor=gray,
            linkcolor=gray,
            urlcolor=black,
            bookmarksopen=true,
            pdfstartview={XYZ null null 1.50},
            pdfpagelayout={SinglePage}
            }

\newcommand{\bR}{\mathbb{R}}

\newcommand{\cH}{\mathcal{H}}

\newcommand{\cK}{\mathcal{K}}
\newcommand{\cN}{\mathcal{N}}

\newcommand{\cT}{\mathcal{T}}
\DeclareMathOperator*{\argmax}{arg\,max}
\DeclareMathOperator*{\st}{s.t.}

\usepackage{newcent}

\usepackage{lineno} 

\begin{document}
\title{\sf Effort Discrimination and Curvature of Contest Technology in Conflict Networks\footnote{Sun, Xu and Zhou are co-first authors. We thank the editor, an advisory editor, two anonymous referees as well as Marcin Dziubiński, Qiang Fu, Jingfeng Lu, John Vickers, Zenan Wu, Yves Zenou and seminar participants for their comments and useful suggestions. This research was supported in part by NSFC (72122017, 72293582, 71773087, 72073083, 72192842), Fok Ying Tong Education Foundation (171076), and  Tsinghua Strategy for Heightening Arts, Humanities and Social Sciences: “Plateaus \& Peaks” (No. 2022TSG08102).
The usual disclaimers apply.}}
\author{
Xiang~Sun\thanks{CEDR and IAS, Economics and Management School, Wuhan University, China. E-mail: \href{mailto:xiangsun.econ@gmail.com}{xiangsun.econ@gmail.com}.}
\and
Jin Xu\thanks{School of Economics, Shandong University, China. E-mail: \href{mailto:jinxu@sdu.edu.cn}{jinxu@sdu.edu.cn}.}
\and
Junjie Zhou\thanks{School of Economics and Management, Tsinghua University, China. E-mail: \href{mailto:zhoujj03001@gmail.com}{zhoujj03001@gmail.com}.}
}
\date{\today}
\maketitle

\begin{abstract}
In a model of interconnected conflicts on a network, we compare the equilibrium effort profiles and payoffs under two scenarios: uniform effort (UE) in which each contestant is restricted to exert the same effort across all the battles she participates, and discriminatory effort (DE) in which such a restriction is lifted. When the contest technology in each battle is of Tullock form, a surprising neutrality result holds within the class of semi-symmetric conflict network structures: both the aggregate actions and equilibrium payoffs under two regimes are the same. We also show that, in some sense, the Tullock form is necessary for such a neutrality result. Moving beyond the Tullock family, we further demonstrate how the curvature of contest technology shapes the welfare and effort effects.


\end{abstract}

\clearpage

\section{Introduction}
\label{sec:intro}

The structure of interaction or relation---visually represented as a network---has become increasingly important in shaping individually strategic choices, resulting in numerous studies of games played on networks within the last few years. Classical studies have mainly concentrated on the network games with linear best-replies and on the network games of strategic complements and substitutes; see, for instance, \cite{ballester2006s}, \cite{bramoulle2007}, \cite{galeotti2010}, etc. Those papers typically explore how network structure impacts equilibrium behavior in various settings.

Conflicts over networks are a class of network games, where contestants can simultaneously participate in multiple battles with various valuations and sizes. The multi-battle relationships can be conveniently modeled as a network, allowing complex conflictual relationships, beyond the traditional studies on contests without network structures. For instance, the leading high tech companies, such as Google, Apple, and Microsoft, invest a significant amount of resources into research and development (R\&D) on the internet markets, which comprises the basis for achieving competitive advantages over competitors. The firms' product range, to which their R\&D is dedicated, is relatively wide, including operating systems, AI, browsers, search engines, cloud services, etc. The strategic interaction among multiple competitors within multi-market(-product) can be conveniently analyzed using a network approach.

In this paper, we extend the framework of \cite{xu2022networks} and consider a conflict model in which {players simultaneously participate in multiple battles, and the valuations of these battles are dependent on their respective sizes}. A contestant's winning probability of a particular battle is specified by a logit contest success function. We compare two policy scenarios: uniform effort (UE) in which each contestant is restricted to exert the same effort across all the battles she participates, and discriminatory effort (DE) in which such a restriction is lifted. As a leading example, we consider marketing strategy of competitive multi-market companies. When advertising through traditional media, companies may struggle to deliver targeted advertising to specific submarkets. In this scenario, the company's standardized marketing strategy can be seen as an instance of uniform effort. However, with the advancement of information technology and the  availability  of data, companies now have the capability to develop customized marketing strategies for different submarkets, which can be regarded as an instance of discriminatory effort.\footnote{Conceptually, the comparison between DE and UE is related to the literature on (third-degree) price discrimination: DE is similar to charging differential prices in different market segments, while UE can be seen as limiting to a uniform pricing.}

In both scenarios of DE and UE, we first fully characterize the unique equilibrium effort profiles and payoffs in Propositions~\ref{prop:de-eqm} and \ref{prop:ue-eqm}, respectively. Under DE, the equilibrium is uniquely determined by the corresponding first order conditions (FOCs) where each contestant balances the marginal costs and marginal benefits across battles. Although these FOCs are highly nonlinear objectives, the solution to the FOC system is unique using the argument in \cite{xu2022networks}. Under UE, the corresponding FOCs reflect the constraints imposed by UE. Uniqueness can be similarly obtained. To make progress, we focus primarily on semi-symmetric conflict networks, in which each contestant engages in the same number of the battles with the same size, and the contest production functions and valuations depend on the size of battle. Several concrete examples of semi-symmetric conflict networks are given in Sections~\ref{sec:example} and \ref{sec:model}. Within this class of conflict structures, we obtain sharper equilibrium characterizations. In particular, the equilibrium under either scheme is interior and symmetric across players. Moreover, the equilibrium in DE is also shown to be semi-symmetric in the sense that each contestant exerts the same effort in battles of the same size.

To address the effect of effort discrimination, we compare aggregate actions and equilibrium payoffs between UE and DE. The comparative exercise is closely related to the production function $f$ in the logit contest success function and the inverse of its semi-elasticity $h=f/f'$.\footnote{Such a function has been considered in literature on contest; see, for example, \cite{fu2009beauty}.} When the contest success function is of Tullock form,\footnote{The production function $f$ is of the power form if and only if $h=f/f\rq{}$ is linear.} then a surprising neutrality result holds within the class of semi-symmetric conflict network structures: both the aggregate actions and equilibrium payoffs for each player under two regimes are the same. Moving beyond the Tullock form, the curvature of contest technology $h$ shapes the welfare and effort effects. More precisely, if $h$ is strictly convex (resp. concave), then DE has a lower (resp. higher) total effort and a higher (resp. lower) expected payoff than UE for each player. To obtain this result, we apply Jensen\rq{}s inequality to a set of reorganized equilibrium conditions. When neutrality does not hold, the choice between UE and DE may serve as a new instrument for contest designers.

We also show that the Tullock form for contest success function or the linearity of $h$ is also necessary for the neutrality result of effort discrimination; see Theorem~\ref{thm:neutrality-tullock}. A major step in the proof of Theorem~\ref{thm:neutrality-tullock} is constructing appropriate variations in battle valuations to prove that, under the neutrality of effort discrimination, $h$ must satisfy Cauchy's equation $h(z_1) + h(z_2) = h(z_1 + z_2)$. Then it is straightforward to see that $h$ is linear and the contest success function is of Tullock form. Thus, the neutrality of effort discrimination and the curvature of conflict technology in our setting are closely related. We also discuss several model extensions.

The literature on multi-battle contests mainly study how contestants allocate resources in multiple battles in a competitive environment; see \cite{kovenock2012contests} for a comprehensive survey. Vairous issues are addressed: the linkages across battles (objective or cost), the timing of moves (simultaneously or sequentially), the information structure (complete or incomplete), the cost function form (quadratic or budget constraint) and contest success functions (all-pay or Tullock); see, for instance, \cite{roberson2006colonel}, \cite{kvasov2007contests}, \cite{konrad2009multi}, \cite{fu2015team}, \cite{kovenock2021generalizations}, \cite{chowdhury2021focality}. In these works, the network often, but not always, takes a simple form where every battle involves all the participants. In our paper, we focus on semi-symmetric structures and illustrates the effort comparison between DE and UE.

Our paper builds on the recent but growing literature that studies equilibrium outcomes in network contests; see \cite{Goyal2016conflictsurvey} for a recent survey. The network characterizes players' social relations in society, so the network structure affects the level of effort of participants in different contests.\footnote{See \cite{jacksonzenou2015}, \cite{vigier2014}, \cite{jackson2015networks}, \cite{franke2015conflict}, \cite{bimpikis2016competitive}, \cite{Hiller2017}, \cite{konig2017}, \cite{Kovenock2018}, \cite{dziubinski2019strategy}, \cite{rietzke2017contests}, for example, all of which have a different focus than the present study and use specific forms.} \cite{franke2015conflict} and \cite{huremovic2021noncooperative} consider conflict networks where multiple participants are involved in multiple bilateral conflicts. \cite{xu2022networks} use variational inequality techniques to address equilibrium uniqueness and propagation of shocks in conflict networks. Typically in these models, a closed-form solution is not available, unless the network structure is very specific and players are symmetric both in terms of network positions and cost functions. \cite{konig2017} consider a single Tullock contest with positive (negative) spillovers by friends (enemies) in order to derive closed-form solutions; these solutions enable the structural estimation of a model in the context of the Great War of Africa. To obtain closed-form equilibrium solutions, \cite{rietzke2017contests} study special families of networks such as biregular graphs and stars with linear cost functions. A central feature of our modeling framework is that although the contest structure is symmetric among players, each player has to compete in battles with {various} sizes and every battle may involve part of all participants. The literature that explores the closed-form solution of individual effort on semi-symmetric network is relatively sparse. The present paper is also closely related to \cite{bimpikis2016competitive}, in which they examine a model of competition between firms that can target their marketing budgets to individuals embedded in a social network. They find that it is optimal for the firms to asymmetrically (discriminatorily) target a subset of the individuals under certain conditions. Our study attempts to provide a comprehensive answer about effects of effort discrimination, which are typically not addressed in these papers.

The comparison between DE and UE in our context bears some similarity to the literature on third-degree price discrimination; see \cite{varian1985price}, \cite{holmes1989effects}, \cite{corts1998third}, \cite{aguirre2010monopoly}, \cite{bbm2015}, \cite{bergemann2022third}, among others. As shown in the latter literature, price discrimination, if without further conditions on primitives, often has \emph{ambiguous} welfare and output effects. For example, \cite{aguirre2010monopoly} use curvature information of demand functions to derive sufficient conditions for discrimination to have positive or negative effects on social welfare and output. More strongly, \cite{bbm2015} use information design techniques to obtain the \emph{surplus triangle} result in the monopolistic setting. For comparison, we obtain an interesting neutrality result of effort discrimination on both welfare and total effort when the contest success functions take the Tullock form. In our setting, the curvature of contest technology $h=f/f'$ plays a critical role in shaping the welfare and effort effects, which is parallel to the demand curvature approach in showing the effects of price discrimination \citep{aguirre2010monopoly}.

The remainder of the paper is organized as follows. In Section~\ref{sec:example}, we present a motivating example, demonstrating that the effects of UE and DE on efforts and welfare relate to the curvature of function $h$. In Sections~\ref{sec:model} and \ref{sec:eqm}, we introduce model and provide the equilibrium analysis under both DE and UE. In particular, we establish the critical role of the curvature of $h$ in shaping the effects of effort discrimination, i.e., the comparisons between DE and UE in terms of equilibrium actions and payoffs. In Section~\ref{sec:neutrality}, we analyze neutrality and provide several discussions. All technical proofs are relegated in Appendix~\ref{sec:appendix}.

\section{A motivating example}
\label{sec:example}

There are three players $\{1, 2, 3\}$ and four battles $\{a, b, c, d\}$ in the conflict network represented by Figure~\ref{fig-triangle}. The details of each battle are given by Table~\ref{table-triangle}, where $v_2$ and $v_3$ denote the prizes for size-$2$ and size-$3$ battles, respectively.

\begin{figure}[htb!]
\begin{floatrow}
\hspace*{15pt}
\capbtabbox{%
\renewcommand{\arraystretch}{1.2}
\begin{tabular}{c||c|c}
Battle  & Participating players   & Prize \\\hline
$a$     & 1, 2      & $v_2 = 5$ \\
$b$     & 2, 3      & $v_2 = 5$ \\
$c$     & 3, 1      & $v_2 = 5$	\\
$d$     & 1, 2, 3   & $v_3 = 72$
\end{tabular}
\vspace*{5pt}
}{%
  \caption{Triangle conflict}%
  \label{table-triangle}
}
\ffigbox{%
  \begin{tikzpicture}[scale=.7]
\fill[green!60!white, opacity=0.4] (90:2cm)--(210:2cm)--(330:2cm)--cycle;
\draw [line width=3.5pt][red!80!white] (90:2cm) --node[above, pos=0.5, sloped]{$a$} (210:2cm) --node[below, pos=0.5, sloped]{$b$} (330:2cm) --node[above, pos=0.5, sloped]{$c$} cycle;
\node (1) at (90:2.02cm) {};
\node (2) at (210:2.02cm) {};
\node (3) at (330:2.02cm) {};
\node at (0,0) {$d$};
\fill (1) circle (0.1) node [above] {1};
\fill (2) circle (0.1) node [left] {2};
\fill (3) circle (0.1) node [right] {3};
\end{tikzpicture}
}{%
  \caption{Triangle conflict}%
  \label{fig-triangle}
}
\end{floatrow}
\end{figure}
\FloatBarrier

All the battles have logit form contest success functions, which admit a common contest production function $f$. We further assume that all players have the same quadratic cost function. For instance, player 1 participates in battles $a$, $c$ and $d$, and her expected payoff is
$$ v_2 \cdot \frac{f(x_1^a)}{f(x_1^a) + f(x_2^a)} + v_2 \cdot \frac{f(x_1^c)}{f(x_1^c) + f(x_3^c)} + v_3 \cdot \frac{f(x_1^d)}{f(x_1^d) + f(x_2^d) + f(x_3^d)} - \frac{1}{2} \bigl( x_1^a + x_1^c + x_1^d \bigr)^2, $$
where each $x_i^t$ is the effort player $i$ exerts in battle $t$.

The structure of the conflict network in Figure~\ref{fig-triangle} is symmetric. It is indeed a particular illustration of the semi-symmetric conflict network, which is formally defined later. We examine the (symmetric) equilibrium efforts (and payoffs) under two scenarios: the scenario of uniform effort (UE) in which each player is restricted to exert the same effort across all the battles she participates ($x_i^t=x_i^{t\rq{}}$ whenever $i$ participates in both battles $t$ and $t\rq{}$), and the scenario of discriminatory effort (DE) where players are allowed to exert different efforts across battles they participate. The analysis is conducted by the first order approach.

In this example we consider the following three forms of the production function $f$:
$$ f_1(x) = \frac{x}{x+1}, \ \ f_2(x) = 2 x^{\frac{1}{2}}, \ \ f_3(x) =
\begin{cases}
2 x^{\frac{1}{2}},	& \text{ if } x \le 1,	\\
x + 1,				& \text{ if } x > 1.
\end{cases}$$
Notice that each of them is an increasing and concave function with $f(0) = 0$. The following table summarizes the equilibrium total efforts of each player, for the three distinct production functions.

\begin{table}[htb!]
\centering
\renewcommand{\arraystretch}{1.2}
\begin{tabular}{c|c||ccc}
$f$			& $h = f/f'$\footnotemark{}	& Total effort under UE	&	& Total effort under DE	\\[2pt]\hline
$f_1(x)$	& convex		& 3.03304	& $>$	& 2.68415	\\
$f_2(x)$	& linear		& 3.04138	& $=$ & 3.04138	\\
$f_3(x)$	& concave		& 3.05522	& $<$	& 3.6833
\end{tabular}
\end{table}
\FloatBarrier
\footnotetext{The three corresponding $h$ functions are: $h_1(x) = x(1+x)$, $h_2(x) = 2x$, and $h_3(x) = \begin{cases}2x,&\text{if }x \le 1,\\ 1+x,&\text{if }x>1.\end{cases}$}

It is shown that the function $h := {f}/{f'}$ plays an important role in characterizing the equilibrium efforts under both scenarios. One may conjecture that the convexity (resp. concavity) of $h$ is a necessary and sufficient condition for the statement that the equilibrium total effort for each player under UE is higher (resp. lower) than that under DE. In each battle, the participants have the same probability of winning under symmetric equilibria. Thus, for each player, the higher total effort exerts, the lower benefit received. Hence, each player will have a lower (resp. higher) expected payoff under UE when the production function is $f_1$ (resp. $f_3$), and each player has the same expected payoff under UE and DE when the production function is $f_2$. So one may also conjecture that the curvature of $h$ is closely related to the comparison on player benefits between DE and UE. We formally address these issues below.

\section{Model}
\label{sec:model}

In this section, we introduce DE and UE after presenting a general model of conflict network.

\vspace{-6pt}
\paragraph{Players and battles} There are $N$ \emph{risk-neutral} players competing in $T$ different battles. The set of players is denoted by $\cN$ and players are indexed by $i = 1, 2, \ldots, N$. The set of battles is denoted by $\cT$. Both $N \ge 2$ and $T \ge 1$ are assumed to be finite.

\vspace{-6pt}
\paragraph{Conflict structure} The conflict structure is modeled by a network, which can be represented by an $N \times T$ matrix $\bm{\Gamma} = (\gamma_i^t)$: $\gamma_i^t = 1$ if player $i$ participates in the battle $t$, otherwise $\gamma_i^t = 0$. Let $\cN^t = \{ i \in \cN \mid \gamma_i^t = 1 \}$ denote the set of participants in battle $t$ and let $n^t = |\cN^t| = \sum_{i \in \cN} \gamma_i^t$ denote its size. Let $\cT_i = \{ t \in \cT \mid \gamma_i^t = 1 \}$ denote the set of battles that player $i$ attends and let $t_i = |\cT_i| = \sum_{t \in \cT} \gamma_i^t$ denote its cardinality.\footnote{Without loss of generality, we assume that the conflict structure does not include any dummy players or battles; that is, $n^t \ge 2$ for each $t \in \cT$ and $t_i \ge 1$ for each $i \in \cN$.} For each player $i$ and battle $t \in \cT_i$, if a player exerts zero effort in battle $t$, it is equivalent to not participating in battle $t$ altogether. Therefore, we allow players to self-select the battles they wish to participate in by either exerting positive effort or not.\footnote{We thank an anonymous referee for suggesting this point.}

The conflict structure is assumed to be semi-regular: every player takes part in the same number of battles with the same size. Formally, there exists a vector $\bm{d} = (d_2, \ldots, d_N)$ such that for each player $i \in \cN$, the number of size-$k$ battles that player $i$ participates in is always the number $d_k$, i.e., $\bigl| \{ t \in \cT_i \mid n^t = k \} \bigr| = d_k$. Let $\cK = \left\{ k \mid n^t = k \text{ for some battle $t$} \right\}$ denote the set of all possible sizes of battles.

\vspace{-6pt}
\paragraph{Conflict technology} In each $t \in \cT$, let $\bm{x}^t = (x_i^t)_{i \in \cN^t} \in \bR_+^{n^t}$ denote the effort vector of all the players participating in the battle $t$. For each battle $t$ in which player $i$ participates, her winning probability is determined by a logit form contest success function (CSF):
\begin{equation}
\label{eq:CSF}
p^t_i(\bm{x}^t) = \frac{f(x_i^t)}{\sum\limits_{j \in \cN^t} f(x_j^t)},\footnote{In the case that $\bm{x}^t = \bm{0}$, the winning probability $p^t_i(\bm{x}^t)$ is defined to be $\frac{1}{|\cN^t|} = \frac{1}{n^t}$.}
\end{equation}
where $f$ is the common contest production function of all battles, satisfying the conditions:
$f(0) = 0$ and for all $x > 0$, $f'(x) > 0$ and $f''(x) \le 0$.\footnote{This logit form of CSF is widely used in modeling contests and conflicts; see, for example, \cite{konrad2009strategy, franke2015conflict, konig2017}.}

For notational simplicity, we use $h$ to denote the inverse of the semi-elasticity of production function, i.e., $h = \frac{f}{f'}$. It is straightforward to verify that $h$ is strictly increasing in $(0, +\infty)$, $\lim\limits_{x \to 0+} h(x) = 0$, and $\lim\limits_{x \to +\infty} h(x) = +\infty$; see Lemma~\ref{lem:h} in Appendix~\ref{sec:appendix}.

\vspace{-6pt}
\paragraph{Valuation, cost, and payoff} In each battle $t$ with size $k$, the winning player obtains an exogenous prize $v^t = v_k > 0$ and others receive nothing.

For each player $i$, exerting efforts $\bm{x}_i = (x_i^t)_{t \in \cT_i}$ induces a cost $C(X_i)$, where $X_i = \sum_{t \in \cT_i} x_i^t$ denotes player $i$'s total effort in all battles that she takes part in. The cost function $C(\cdot) \colon \bR_+ \to \bR_+$ is assumed to be twice continuously differentiable, strictly increasing, and convex.

Thus, the expected payoff of each player $i$ is given by
\begin{equation}
\label{eq:payoff}
\Pi_i(\bm{x}_i, \bm{x}_{-i}) = \sum_{t \in \cT_i} v^t \cdot \frac{f(x_i^t)}{\sum\limits_{j \in \cN^t} f(x_j^t)} - C\Bigl(\sum_{t \in \cT_i} x_i^t\Bigr).
\end{equation}
In other words, payoffs are dependent on the sum of the battle values weighted by the corresponding winning probabilities minus the effort cost.

We have described a \emph{semi-symmetric conflict network} as a tuple $\bigl( \cN, \cT, \bm{\Gamma}, f(\cdot), (v_k)_{k\in\cK}, C(\cdot) \bigr)$. It is clear that the triangle conflict in Section~\ref{sec:example} is an example of semi-symmetric conflict network.

\vspace{-6pt}
\paragraph{UE and DE}
In a semi-symmetric conflict network, we shall consider the equilibrium efforts and payoffs under two scenarios: the scenario of uniform effort (UE) where each player is restricted to exert the same effort across all the participating battles, and the scenario of discriminatory effort (DE) in which such a restriction is lifted.

We slightly abuse the terminology by using DE and UE to represent the corresponding games.

\section{Equilibrium analysis}
\label{sec:eqm}

Given a semi-symmetric conflict network, we first investigate the equilibrium in DE. Under this scenario, we use $\bm{x}_i = (x_i^t)_{t \in \cT_i} \in \bR_+^{t_i}$ to denote a strategy of player $i$. A strategy profile $\bm{x} = (\bm{x}_i)_{i\in\cN}$ is said to be \emph{semi-symmetric} if there exists an effort vector $(x_k)_{k\in\cK}$ such that in every size-$k$ battle $t$, each involved player exerts the same effort $x_k$, i.e., $x_i^t = x_k$ for each $i\in\cN^t$. In other words, a semi-symmetric strategy profile requires that the effort each player exerts in a battle is size-determined. Alternatively, a semi-symmetric strategy profile can be represented by the corresponding effort vector $(x_k)_{k\in\cK}$.

We revisit the triangle conflict in Section~\ref{sec:example}. We denote a strategy of player 1 as $(x_1^a, x_1^c, x_1^d)$, a strategy of player 2 as $(x_2^a, x_2^b, x_2^d)$, and a strategy of player 3 as $(x_3^b, x_3^c, x_3^d)$, where the superscripts indicate the corresponding battles. Since $a$, $b$ and $c$ are all size-2 battles, the semi-symmetry on strategy profile requires that $x_1^a = x_2^a = x_2^b = x_3^b = x_3^c = x_1^c$. Analogously, the semi-symmetry also implies that $x_1^d = x_2^d = x_3^d$.

\begin{prop}
\label{prop:de-eqm}
For each semi-symmetric conflict network $\bigl( \cN, \cT, \bm{\Gamma}, f(\cdot), (v_k)_{k \in \cK}, C(\cdot) \bigr)$, there is a unique Nash equilibrium $\bm{x}^*$ under the scenario of discriminatory effort. Furthermore, $\bm{x}^*$ is semi-symmetric and interior. In particular, in this Nash equilibrium $\bm{x}^* = (x^*_k)_{k \in \cK}$, for each $k \in \cK$, the effort $x^*_k$ exerted in each size-$k$ battle satisfies
\begin{equation}
\label{eq:de-FOC}
v_k \cdot \frac{k-1}{k^2} \cdot \frac{f'(x^*_k)}{f(x^*_k)} = \lambda^*,
\end{equation}
where
\begin{equation}
\label{eq:de-MC}
\lambda^*  =  C' \Big( \sum_{\ell \in \cK} d_\ell x^*_\ell \Big)
\end{equation}
is the marginal cost in equilibrium.
\end{prop}

Since $f(\cdot)$ is strictly increasing and concave, and $C(\cdot)$ is twice continuously differentiable, increasing, and convex, Theorem~1 in \cite{xu2022networks} guarantees the existence of Nash equilibria. Moreover, since $C(\cdot)$ is also strictly increasing, Theorem~2(ii) in \cite{xu2022networks} implies that DE admits a unique Nash equilibrium. The complete proof of Proposition~\ref{prop:de-eqm} is given in Appendix~\ref{sec:appendix}.

We revisit the triangle conflict in Section~\ref{sec:example}. Suppose the production function is $f$. Since $\cK = \{2, 3\}$, $d_2 = 2$, $d_3 = 1$, and $C(X) = \frac{1}{2}X^2$, Proposition~\ref{prop:de-eqm} implies that the equilibrium efforts $(x^*_2, x^*_3)$ in DE are characterized by the following first order conditions:
$$ v_2 \cdot \frac{1}{4} \cdot \frac{1}{h(x^*_2)} = \lambda^*, \quad v_3 \cdot \frac{2}{9} \cdot \frac{1}{h(x^*_3)} = \lambda^*, $$
where $x^*_2$ and $x^*_3$ are individual efforts exerting in size-2 battles and in size-3 battles, respectively, $\lambda^* = 2 x^*_2 + x^*_3 = X^*$ is the marginal cost,  {and $X^*$ is the total effort exerted by each player under DE}.

\begin{remark}
\rm
\label{rmd:de-eqm}
Let $\bigl( \cN, \cT, \bm{\Gamma}, f(\cdot), (v_k)_{k\in\cK}, C(\cdot) \bigr)$ be a semi-symmetric conflict network. Suppose $\hat{\bm{x}}$ is an interior semi-symmetric strategy profile with the effort vector $(\hat{x}_k)_{k\in\cK}$ (not necessarily a semi-symmetric Nash equilibrium). Then one can find valuations $(\hat{v}_k)_{k\in\cK}$ such that $\hat{\bm{x}}$ is the unique Nash equilibrium of the new semi-symmetric conflict network $\bigl( \cN, \cT, \bm{\Gamma}, f(\cdot), (\hat{v}_k)_{k\in\cK}, C(\cdot) \bigr)$ under DE.

To be more precise, given the conflict structure $\bm{\Gamma}$, the production function $f(\cdot)$, the cost function $C(\cdot)$, and the interior semi-symmetric strategy profile $\hat{\bm{x}}$, let
$$ \hat{v}_k = \frac{k^2}{k-1} \cdot \frac{f(\hat{x}_k)}{f'(\hat{x}_k)} \cdot C' \Bigl( \sum_{\ell \in \cK} d_\ell \hat{x}_\ell \Bigr) \text{ for each $k\in\cK$}. $$
Then Equations~\eqref{eq:de-FOC} and~\eqref{eq:de-MC} hold for $(\hat{x}_k)_{k\in\cK}$ and $(\hat{v}_k)_{k\in\cK}$. Thus, by Proposition~\ref{prop:de-eqm}, the given semi-symmetric strategy profile $\hat{\bm{x}}$ is the unique Nash equilibrium of the new conflict network $\bigl( \cN, \cT, \bm{\Gamma}, f(\cdot), (\hat{v}_k)_{k\in\cK}, C(\cdot) \bigr)$ under DE.
\end{remark}

\medskip

In the rest of this section, we consider the other scenario, which further requires that each player can only set a uniform effort level which is the same across all involved battles of her. A typical strategy for each player $i$ is to choose a single effort level $x_i$, so that $x_i^t = x_i^{t'} = x_i$ for all involved battles $t$ and $t'$. When all players adopt uniform efforts, each player $i$'s payoff function becomes
$$ \Pi^u_i(x_i, x_{-i})  =  \sum_{t \in \cT_i} v^t \cdot \frac{f(x_i)}{\sum\limits_{j \in \cN^t} f(x_j)} - C \Big( \sum_{\ell \in \cK} d_\ell x_i \Big). $$
We have the following equilibrium result.

\begin{prop}
\label{prop:ue-eqm}
For each semi-symmetric conflict network $\bigl( \cN, \cT, \bm{\Gamma}, f(\cdot), (v_k)_{k \in \cK}, C(\cdot) \bigr)$, there is a unique Nash equilibrium $\bm{x}^u = (x^u, x^u, \ldots, x^u)$ under the scenario of uniform effort:
\begin{equation}
\label{eq:ue-FOC}
\sum_{\ell \in \cK} d_\ell \cdot v_\ell \cdot \frac{\ell-1}{\ell^2} \cdot \frac{f'(x^u)}{f(x^u)} = \lambda^u,
\end{equation}
where
\begin{equation}
\label{eq:ue-MC}
\lambda^u  =  C' \Big( \sum_{\ell\in\cK} d_\ell x^u \Big) \cdot \Big( \sum_{\ell\in\cK} d_\ell \Big)
\end{equation}
is the marginal cost in equilibrium.
\end{prop}

\noindent The uniqueness of Nash equilibrium follows Proposition~5 in \cite{xu2022networks}. The complete proof of Proposition~\ref{prop:ue-eqm} is given in Appendix~\ref{sec:appendix}.

In the triangle conflict in Section \ref{sec:example}, the equilibrium effort $x^u$ in UE is characterized by the following first order condition:
$$ \frac{v_2}{2} \cdot \frac{1}{h(x^u)} + \frac{2v_3}{9} \cdot \frac{1}{h(x^u)} = \lambda^u, $$
where $\lambda^u = (d_2 x^u + d_3 x^u)(d_2 + d_3) = 9 x^u = 3 X^u$ is the marginal cost, {and $X^u$ is the total effort exerted by each player under UE}.

\medskip

The following theorem establishes a neat correspondence between the curvature of function $h = \frac{f}{f'}$ and the size relationship of equilibrium total efforts under DE and UE. It provides an affirmative answer for the conjecture in Section \ref{sec:example}.

\begin{thm}
\label{thm:comparison-h}
For any semi-symmetric conflict network and each player involved, (1) the total effort in DE does not exceed that in UE if $h$ is convex; (2) the total effort in DE is not less than that in UE if $h$ is concave; (3) DE and UE have the same total effort if $h$ is linear.\footnote{It is easy to see that $h'' = \frac{2ff''f''-f'f'f''-ff'f'''}{(f')^3}$. Since $f'>0$, $h$ is convex (resp. concave) if and only if $2ff''f''-f'f'f''-ff'f''' \ge 0$ (resp. $\le 0$). In a Tullock contest, the production function is $f(x) = x^r$ for some $r > 0$. Then we know $h(x) = \frac{f(x)}{f'(x)} = \frac{x}{r}$, which is linear. In a Hirshleifer contest, the production function is $f(x) = e^{\alpha x}$ for some $\alpha > 0$. The function $h(x)$ is $h(x) = \frac{1}{\alpha}$, which is a constant. If the production function is a CARA utility $f(x) = 1 - e^{-\alpha x}$ for some $\alpha > 0$, then $h(x) = \frac{1}{\alpha} (e^{\alpha x} -1)$, which is convex.}
\end{thm}

We revisit the triangle conflict in Section \ref{sec:example} to illustrate the results in Theorem~\ref{thm:comparison-h} under a general production function $f$. Recall that the equilibrium efforts $(x^*_2, x^*_3)$ in DE and $x^u$ in UE are respectively characterized by the following first order conditions
\begin{minipage}[c]{0.395\textwidth}
\begin{align}[left=\empheqlbrace]
v_2 \cdot \tfrac{1}{4} \cdot \tfrac{1}{h(x^*_2)}	& = \lambda^*,	\label{eq:DE1}\\
v_3 \cdot \tfrac{2}{9} \cdot \tfrac{1}{h(x^*_3)}	& = \lambda^*,	\label{eq:DE2}\\
2 x^*_2 + x^*_3	& = \lambda^*.	\label{eq:DElambda}
\end{align}
\end{minipage}
\qquad
\begin{minipage}[c]{0.55\textwidth}
\begin{align}[left=\empheqlbrace]
\tfrac{v_2}{2} \cdot \tfrac{1}{h(x^u)} + \tfrac{2v_3}{9} \cdot \tfrac{1}{h(x^u)}	& = \lambda^u,	\label{eq:UE}\\
3(2 x^u + x^u)	& = \lambda^u.	\label{eq:UElambda}
\end{align}
\end{minipage}
According to the shape, especially the curvature of contest technology, $h$, we divide the discussion into two parts:
\begin{enumerate}[label=(\roman*)]
\item
Suppose that $h$ is a linear function, say $h(z) = \frac{z}{r}$ for some $r > 0$. By solving Equations~\eqref{eq:DE1}, \eqref{eq:DE2} and \eqref{eq:DElambda}, we have
$$ x^*_2 = \frac{\sqrt{9r}v_2}{\sqrt{8(9v_2+4v_3)}}, \ x^*_3 = \frac{\sqrt{8r}v_3}{\sqrt{9(9v_2+4v_3)}}. $$
Thus,
$$ X^* = 2 x^*_2 + x^*_3 = \sqrt{\frac{v_2r}{2}+\frac{2v_3r}{9}}. $$
On the other hand, from Equations~\eqref{eq:UE} and \eqref{eq:UElambda}, we know that
$$ 2 \frac{v_2}{4} \frac{r}{x^u} + \frac{2v_3}{9} \frac{r}{x^u} = \lambda^u = 9 x^u, $$
and hence
$$ X^u = 3 x^u = \sqrt{\frac{v_2r}{2}+\frac{2v_3r}{9}}, $$
which is the same as $X^*$.

\item
Next we consider the case when $h$ is convex.\footnote{The case with concave $h$ is similarly discussed.} We prove by contradiction. Suppose that DE has a higher total effort (i.e., $X^* > X^u$). From Equations~\eqref{eq:DE1} and \eqref{eq:DE2}, we have
\begin{align}
v_2 \cdot \frac{1}{4} \cdot \frac{1}{X^*}	& = h(x_2^*),	\label{eq:DE3}\\
v_3 \cdot \frac{2}{9} \cdot \frac{1}{X^*}	& = h(x_3^*).	\label{eq:DE4}
\end{align}
Summing Equations~\eqref{eq:DE3} and \eqref{eq:DE4} with respective weights $\frac{2}{3}$ and $\frac{1}{3}$, we have
$$ \frac{1}{X^*} \cdot \Big( \frac{v_2}{6} + \frac{2v_3}{27} \Big) = \frac{2}{3} h(x^*_2) + \frac{1}{3} h(x^*_3). $$
When $h$ is strictly convex, Jensen's inequality implies
$$ \frac{2}{3} h(x^*_2) + \frac{1}{3} h(x^*_3) > h \Big( \frac{2x^*_2+x^*_3}{3} \Big) > h(x^u), $$
where the last inequality follows from strict monotonicity of $h$. From Equation~\eqref{eq:UE}, we have $h(x^u) = \frac{1}{X^u} ( \frac{v_2}{6} + \frac{2v_3}{27} )$. Thus, $X^u > X^*$, which leads to a contradiction. Therefore, DE has a lower total effort when $h$ is strictly convex.
\end{enumerate}

\strut

Theorem~\ref{thm:comparison-h} follows from a similar argument: Based on Propositions~\ref{prop:de-eqm} and \ref{prop:ue-eqm}, we have the following equations:
\begin{align*}
\sum_{k \in \cK} d_k \cdot v_k \cdot \frac{k-1}{k^2} \cdot \frac{1}{\lambda^*}    & = \sum_{k \in \cK} d_k \cdot h(x^*_k),  \\
\sum_{k \in \cK} d_k \cdot v_k \cdot \frac{k-1}{k^2} \cdot \frac{\sum_\ell d_\ell}{\lambda^u}    & = \sum_{k \in \cK} d_k \cdot h(x^u).
\end{align*}
Suppose that $h$ is convex. We start with the same total efforts for simplicity, i.e., $\sum_k d_k \cdot x^*_k = \sum_k d_k \cdot x^u$. Since $h$ is convex, the discriminatory efforts $(x^*_k)$ make the weighted sum $\sum_k d_k \cdot h(x^*_k)$ larger than $\sum_k d_k \cdot h(x^u)$. It in turn implies that $\lambda^* = C' \big( \sum_\ell d_\ell x^*_\ell \big)$ is less than $\frac{\lambda^u}{\sum_\ell d_\ell} = C' \big( \sum_\ell d_\ell x^u \big)$. Equivalently, $\sum_\ell d_\ell x^*_\ell \le \sum_\ell d_\ell x^u$, i.e., the total effort in DE is smaller than that in UE.

The curvature of $h$ also plays a critical role in \cite{fu2009beauty}, who study how the total effort of contestants changes when a ``grand'' contest is allowed to be split into a set of parallel ``subcontests.'' When $h$ is convex or linear\footnote{See Definition 2 in \cite{fu2009beauty} and the discussion therein.}, \cite{fu2009beauty} show that a grand contest generates more effort than any set of subcontests. The convexity of $h$ (including the linear case) is shown to be a sufficient condition to derive the results in \cite{fu2009beauty}, and it is unknown whether the converse holds when $h$ is strictly concave.\footnote{See, for instance, \cite{fu2012ET,fu2021multi,fu2022JET,fu2023GEB} for recent advances in multi-prize contests.} In our setting, the convexity (concavity) of $h$ is necessary for UE (DE) to generate more effort. Moreover, Theorem~\ref{thm:comparison-h} provides a neutrality result---UE and DE induce the same total effort when $h$ is linear, i.e., both convex and concave.

Since the equilibrium in DE is semi-symmetric and the equilibrium in UE is symmetric, participants in each battle have the same probability of winning. So an individual player has a higher expected payoff if she exerts a lower total effort. Hence, a player has a higher (resp. lower) expected payoff under DE if $h$ is convex (resp. concave).

\section{Neutrality and discussion}
\label{sec:neutrality}

In this section, we will delve deeper into the property of (effort) neutrality, where DE and UE have the same total effort for each player. According to Theorem \ref{thm:comparison-h}, if $h$ is linear, then DE and UE will result in the same total effort for each player. Furthermore, since $h = \frac{f}{f'}$ and $f(0) = 0$, $h$ is a linear function if and only if $f$ is of the power form or $f(x) = x^r$, which is further equivalent to the logit form CSF being Tullock. In other words, the Tullock form of CSF or the linearity of $h$ guarantees the property of neutrality. We will demonstrate that the Tullock form of CSF or the linearity of $h$ is also necessary for neutrality to occur.

\subsection{Neutrality}

We take another look at the triangle conflict. Suppose the neutrality property holds generally for any valuations $v_2$ and $v_3$. In this case, DE and UE result in the same total effort for each player, meaning $X^* = X^u$. By substituting Equations~\eqref{eq:DE1} and \eqref{eq:DE2} into the first and second terms on the left-hand side of Equation~\eqref{eq:UE}, respectively, we obtain:
$$ \frac{2}{3} \cdot \frac{h(x^*_2)}{h(x^u)} + \frac{1}{3} \cdot \frac{h(x^*_3)}{h(x^u)} = 1. $$
Notice that $3 x^u = X^u = X^* = 2 x^*_2 + x^*_3$.

By varying the valuations $v_2$ and $v_3$, we can see from Remark \ref{rmd:de-eqm} that the equation
$$ \frac{2}{3} \cdot \frac{h(\hat{x}_2)}{h(\hat{x})} + \frac{1}{3} \cdot \frac{h(\hat{x}_3)}{h(\hat{x})} = 1 $$
holds for any positive $\hat{x}_2$, $\hat{x}_3$, and $\hat{x} = \frac{2}{3} \hat{x}_2 + \frac{1}{3} \hat{x}_3$.\footnote{For more details, please see the proof of Theorem~\ref{thm:neutrality-tullock}.} Equivalently, we have the equation
$$ \frac{2}{3} h(\hat{x}_2) + \frac{1}{3} h(\hat{x}_3) = h \Big( \frac{2 \hat{x}_2 + \hat{x}_3}{3} \Big) $$
for any positive $\hat{x}_2$ and $\hat{x}_3$. One can verify that the function $h$ satisfies Cauchy's equation $h(z_1) + h(z_2) = h(z_1 + z_2)$ for any positive $z_1$ and $z_2$, which in turn implies that $h$ should be a linear function, say $h(x) = \frac{x}{r}$ for some $r > 0$. It is then equivalent to $f(x) = x^r$.

Till here, we have an observation that Tullock form CSF (or linearity of $h$) is necessary for the generic property of neutrality in the triangle conflict. The formal statement of this result, where {we introduce the flexibility for production functions to be dependent on the sizes of battles}, is as follows.

Let $\cH$ be the collection of all semi-symmetric conflict networks such that (1) the set of players is $\cN$; (2) the set of battles is $\cT$; (3) the conflict structure is $\bm{\Gamma}$; (4) the size-determined production functions are $(f_k)_{k\in\cK}$; and (5) the cost function is $C(\cdot)$.

\begin{defn}
\label{defn:neutrality}
\rm
The collection $\cH$ is said to be (effort) \emph{neutral} if for any semi-symmetric conflict network $H$ in $\cH$, its semi-symmetric Nash equilibrium $\bm{x}^*$ under DE and the symmetric Nash equilibrium $\bm{x}^u$ under UE have the same total effort for each player.
\end{defn}

For each semi-symmetric conflict network $H$ in a neutral collection $\cH$, each player $i$ has the same winning probability in each battle $t$ under the two equilibria $\bm{x}^*$ and $\bm{x}^u$, and hence the same payoff.

\begin{thm}
\label{thm:neutrality-tullock}
Suppose that the collection $\cH$ is neutral. Then the production function $f_k(x)$ should be $x^{r_k}$ for some $r_k \in (0, 1]$.
\end{thm}

Theorem \ref{thm:neutrality-tullock} implies that each of the production function $f_k$ must be of Tullock form. Of course, a special case would be a common $f$ as in our baseline model. The example of triangular conflict in Section~\ref{sec:example} demonstrates why CSF is of Tullock form in a simple environment, in which the production functions are the same for each battle and each player. To handle heterogeneous production functions $f_k$, we construct an auxiliary semi-symmetric conflict network in the proof of Theorem~\ref{thm:neutrality-tullock}. It enables us to show that each $h_k := \frac{f_k}{f_k'}$ satisfies Cauchy's equation. It further implies that each $h_k$ is a linear function and $f_k(x)$ should be $x^{r_k}$. Focusing on a common production function $f$, Theorems~\ref{thm:comparison-h} and \ref{thm:neutrality-tullock} together state that the Tullock form CSF is a necessary and sufficient condition for the neutrality of effort discrimination.

\begin{remark}
A salient feature of Tullock technology is homogeneity of degree zero of the contest success function. Such a homogeneity property also plays a similar role in related studies; for instance, \cite{fu2015team} establish neutrality of temporal structures in a model of team contests with pairwise battles.
\end{remark}

The following example illustrates the tightness of semi-symmetry for the neutrality property.

\begin{example}
We consider a variation of the triangle conflict. Let the common production function be $f(x) = x$ and the common cost function be $C(X) = \frac{1}{2} X^{5/2}$. The conflict structure is the same as the triangle conflict, and the notation for each battle is retained. We suppose that the prizes for battles $b$, $c$, and $d$ are $1$, $1$, and $1.6$, respectively. By varying the prize of battle $a$, player 1's total efforts under UE and DE are collected in the following table. When the prize is not $1$, the neutrality does not hold and the comparison between UE and DE is ambiguous.

\begin{table}[htb!]
\centering
\renewcommand{\arraystretch}{1.2}
\begin{tabular}{c||ccc}
Prize of $a$	& Total effort under UE	&	& Total effort under DE	\\\hline
$0$		& 0.65607	& $<$	& 0.75592	\\
$0.1$	& 0.79202	& $>$	& 0.76728	\\
$1$		& 0.85928	& $=$	& 0.85928	\\
$2$		& 0.92031	& $<$	& 0.94612
\end{tabular}
\end{table}
\FloatBarrier

{Notice that by setting the prize of battle $a$ to be zero, the conflict structure reduces to a non-semi-regular structure, where there are only battles $b$, $c$, and $d$. The following numerical result implies that the neutrality does not necessarily hold if the structure is not semi-regular.}
\begin{table}[htb!]
\centering
\renewcommand{\arraystretch}{1.2}
\begin{tabular}{c||ccc}
			& Total effort under UE	&	& Total effort under DE	\\\hline
Player 1	& 0.65607	& $<$	& 0.75592	\\
Player 3	& 1.03262	& $>$	& 0.83914	\\
All players	& 2.34476	& $<$	& 2.35099
\end{tabular}
\end{table}
\FloatBarrier

\end{example}

\subsection{Discussion}

In this subsection, we will address several related issues, including comparative statics, optimal contest design, and equilibrium analysis for conflict network with pure budget.

\paragraph{Comparative static analysis} We take $C(X) = \frac{1}{\rho} X^\rho$ and $f_k(x) = x^{r_k}$, where $\rho \ge 1$ and each $r_k \in (0, 1]$ is an exogenous parameter for size-$k$ battles.

Using Propositions \ref{prop:de-eqm} and \ref{prop:ue-eqm}, we explicitly compute the equilibrium efforts and the total efforts under both scenarios:
$$ x_k^*  =  v_k r_k \frac{k-1}{k^2} \bigg[ \sum_{\ell\in\cK} d_\ell v_\ell r_\ell \frac{\ell-1}{\ell^2} \bigg]^{\frac{1-\rho}{\rho}}, \quad x^u  =  \frac{1}{\sum_{\ell \in \cK} d_\ell} \bigg[ \sum_{\ell\in\cK} d_\ell v_\ell r_\ell \frac{\ell-1}{\ell^2} \bigg]^{\frac{1}{\rho}}, $$
\begin{equation}
\label{eq:cs}
X^*  =  X^u  =  \bigg[ \sum_{\ell\in\cK} d_\ell v_\ell r_\ell \frac{\ell-1}{\ell^2} \bigg]^{\frac{1}{\rho}}.
\end{equation}
Moreover, the equilibrium payoff for each player $i$ is
$$ \Pi^* = \Pi^u = \sum_{\ell\in\cK} \frac{d_\ell v_\ell}{\ell} - \frac{1}{\rho} \sum_{\ell\in\cK} d_\ell r_\ell v_\ell \frac{\ell-1}{\ell^2}. $$

We then have the following results:
\begin{enumerate}[label=(\roman*)]
\item The property of neutrality holds even when $r_k$ can vary with battles (Theorem \ref{thm:neutrality-tullock}).
\item $X^*$ and $X^u$ are increasing in each $d_\ell$, $v_\ell$ and $r_\ell$ for all $\ell \in \cK$.
\item $\Pi^*$ and $\Pi^u$ are increasing in $d_\ell$ and $v_\ell$ for all $\ell \in \cK$, and decreasing in $r_\ell$ for all $\ell \in \cK$.
\end{enumerate}

These results are straightforward. As $d_\ell$ increases, there are more size-$\ell$ battles and then more total prize to win, and hence the total effort becomes larger.

\paragraph{Optimal design}
When we keep the total prize in the conflict constant, a natural question arises: how can we distribute the total prize among battles to maximize the overall effort level? {We continue to use the parameterized effort cost function $C(X) = \frac{1}{\rho} X^\rho$ and production function $f_k(x) = x^{r_k}$ from the previous part.} According to Equation~\eqref{eq:cs}, it is equivalent to solve the following problem:
$$ \max_{(v_\ell)} \sum_{\ell\in\cK} d_\ell v_\ell r_\ell \frac{\ell-1}{\ell^2}, \ \st \sum_{\ell\in\cK} v_\ell \frac{d_\ell n}{\ell} = V, $$
where $V$ is total prize for all battles. Each player participates in $d_\ell$ of size-$\ell$ battles, and each size-$\ell$ battle is counted exactly $\ell$ times. Thus, $\frac{d_\ell n}{\ell}$ is the total number of battles of size $\ell$.

If $\frac{k-1}{k} r_k < \frac{k'-1}{k'} r_{k'}$ for $k, k' \in \cK$, then zero prize should be assigned to size-$k$ battles. Thus, the designer will allocate the prize only to the battles with the size
$$ \ell^* \in \cK^* = \argmax_{\ell\in\cK} \frac{\ell-1}{\ell} r_\ell. $$
And the maximal total effort of each player is
$$ \bigg[ \sum_{\ell\in\cK^*} d_\ell v_\ell r_\ell \frac{\ell-1}{\ell^2} \bigg]^{\frac{1}{\rho}} = \bigg[ r_{\ell^*} \frac{\ell^*-1}{\ell^*} \sum_{\ell\in\cK^*} d_\ell v_\ell \frac{1}{\ell} \bigg]^{\frac{1}{\rho}} = \bigg[ r_{\ell^*} \frac{\ell^*-1}{\ell^*} \frac{V}{n} \bigg]^{\frac{1}{\rho}}, $$
for any $\ell^*\in\cK^*$.

Since $\frac{\ell-1}{\ell} r_\ell$ is increasing in $r_\ell$ and $\ell$, the optimal design depends on the balance between the $r_\ell$ and the battle size $\ell$. In particular, for homogenous $(r_\ell)$ (i.e., all $r_\ell$ are the same), the designer awards the entire prize budget to the battles with the largest size in $\cK$. Moreover, when $\cK$ is fixed, the maximal total effort does not depend on the profile $(d_\ell)$. That is, the number of battles of the same size has a neutral effect on the overall effort level.

\paragraph{Conflict network with pure budget}
We then turn to another setup where the conflict structure, the prize profile $(v_k)$, and production functions $(f_k=x^{r_k})$ are retained, but the cost function is the pure-budget case. We assume that each player has a budget constraint $B>0$. Each player tries to maximize his expected payoff such that the budget constraint is satisfied.

By the similar arguments of the proof of Proposition~\ref{prop:de-eqm}, we can show that there is a unique Nash equilibrium $\bm{x}^* = (x_k^*)_{k \in \cK}$ under DE, where $x^*_k = B \cdot \frac{v_k r_k \frac{k-1}{k^2}}{\sum_{\ell\in\cK} d_\ell v_\ell r_\ell \frac{\ell-1}{\ell^2}}$. There is also a unique Nash equilibrium $\bm{x}^u = (x^u, x^u, \ldots, x^u)$ under UE, where $x^u = \frac{B}{\sum_{\ell \in \cK} d_\ell}$. As all budgets are fully utilized, this setup also maintains the effort neutrality property.   Payoff neutrality property also holds.

{\small
\bibliographystyle{chicagoa}
\let\oldbibliography\thebibliography
\renewcommand{\thebibliography}[1]{%
  \oldbibliography{#1}%
  \setlength{\itemsep}{1pt}%
}
\bibliography{bib-conflict}
}

\singlespacing
\appendix
\section{Appendix}
\label{sec:appendix}

\subsection{Preliminary result}

Several proofs make use of the following lemma on the function of $h$.

\begin{lem}
\label{lem:h}
Let $f$ be a contest production function satisfying $f(0) = 0$, $f' > 0$, and $f'' \le 0$. Then we have the following results on $h(x) = \frac{f(x)}{f'(x)}$:
\begin{enumerate}
\item $h(x)$ is strictly increasing in $x \in (0, +\infty)$.
\item $\lim\limits_{x \to 0+} h(x) = 0$.
\item $\lim\limits_{x \to +\infty} h(x) = +\infty$.
\end{enumerate}
\end{lem}

\begin{proof}[Proof of Lemma~\ref{lem:h}]
Since $f(0) = 0$ and $f'(x) > 0$ for all $x > 0$, we have that $f(x)$ is strictly increasing in $x \in (0, +\infty)$ and $f(x) > 0$ for all $x > 0$. Moreover, since $f''(x) \le 0$, we have that $f'(x)$ is decreasing in $x \in (0, +\infty)$. Thus, together with the fact $f'(x) > 0$ for all $x > 0$, we have that $h(x) = \frac{f(x)}{f'(x)}$ is positive for all $x > 0$ and is strictly increasing in $x \in (0, +\infty)$.

Fix $x_0 > 0$. Since $f'(x) > 0$ is decreasing in $x \in (0, +\infty)$, we have that $f'(x) \ge f'(x_0) > 0$ for any $x \in (0, x_0]$. Hence, $0 < \frac{f(x)}{f'(x)} \le \frac{f(x)}{f'(x_0)}$ for any $x \in (0, x_0]$. Letting $x \to 0+$, we have $\lim\limits_{x \to 0+} f(x) = 0$, and hence $\lim\limits_{x \to 0+} \frac{f(x)}{f'(x_0)} = 0$. Thus, $\lim\limits_{x \to 0+} h(x) = \lim\limits_{x \to 0+} \frac{f(x)}{f'(x)} = 0$.

Since $f'(x) > 0$ is decreasing in $x \in (0, +\infty)$, $f'(x)$ has a nonnegative lower bound. That is, we have either $\lim\limits_{x \to +\infty} f'(x) = 0$ or $\lim\limits_{x \to +\infty} f'(x) > 0$. (1) If $\lim\limits_{x \to +\infty} f'(x) = 0$, then $f(x)$ converges to a positive constant as $x \to +\infty$, and hence $\lim\limits_{x \to +\infty} h(x) = \lim\limits_{x \to +\infty} \frac{f(x)}{f'(x)} = +\infty$. (2) If $\lim\limits_{x \to +\infty} f'(x) > 0$, then $\lim\limits_{x \to +\infty} f(x) = +\infty$, and hence $\lim\limits_{x \to +\infty} h(x) = \lim\limits_{x \to +\infty} \frac{f(x)}{f'(x)} = +\infty$. So in both cases, we have that $\lim\limits_{x \to +\infty} h(x) = \lim\limits_{x \to +\infty} \frac{f(x)}{f'(x)} = +\infty$.
\end{proof}

\subsection{Proofs of Propositions~\ref{prop:de-eqm} and \ref{prop:ue-eqm}}

\begin{proof}[Proof of Proposition~\ref{prop:de-eqm}]
Let $\bigl( \cN, \cT, \bm{\Gamma}, f(\cdot), (v_k)_{k \in \cK}, C(\cdot) \bigr)$ be a semi-symmetric conflict network. We also adopt another equivalent representation $\bigl( \cN, \cT, \bm{\Gamma}, f(\cdot), (v^t)_{t \in \cT}, C(\cdot) \bigr)$ for convenience. We first construct an interior semi-symmetric Nash equilibrium under DE, and then show it to be unique.

It is easy to see that each payoff function $\Pi_i(\bm{x}_i, \bm{x}_{-i})$ is concave in $\bm{x}_i$. Thus, in equilibrium, each $\bm{x}_i$ should satisfy the first-order conditions---each player $i$'s marginal benefit from exerting increment effort in each battle $t$ must be no more than the marginal cost, with equality if the solution $x_i^t$ is interior. That is, for each battle $t$, we have
$$ v^t \cdot \frac{f'(x_i^t) \cdot \Big[ \sum\limits_{j \in \cN^t, j \ne i} f(x^t_j) \Big]}{\Big[ \sum\limits_{j \in \cN^t} f(x_j^t) \Big]^2} \le C' \Big( \sum_{t \in \cT_i} x_i^t \Big) \text{ with equality if $x_i^t > 0$}. $$

We would like to construct a semi-symmetric strategy profile satisfying the above first-order conditions with equality. That is, we try to provide a solution $(x_k)_{k\in\cK}$ for the system of equations
\begin{align}
v_k \cdot \frac{k-1}{k^2} \cdot \frac{f'(x_k)}{f(x_k)}	& = C'(\mu) \text{ for each $k \in \cK$},	\label{eq:de-FOC1}	\\
\mu														& = \sum_{\ell \in \cK} d_\ell x_\ell.	\label{eq:de-MC1}
\end{align}
For each $\mu > 0$, we have $C'(\mu) > 0$. By Lemma~\ref{lem:h}, $\frac{f'(x_k)}{f(x_k)}$ is strictly decreasing in $x_k \in (0, +\infty)$, $\lim\limits_{x_k \to 0+} \frac{f'(x_k)}{f(x_k)} = +\infty$, and $\lim\limits_{x_k \to +\infty} \frac{f'(x_k)}{f(x_k)} = 0$. Thus, there exists a unique solution for Equation~\eqref{eq:de-FOC1}, denoted by $x_k = g_k(\mu)$. Clearly, $x_k = g_k(\mu) > 0$ for all $\mu > 0$. Since $C(\cdot)$ is strictly increasing, $g_k(\mu)$ is decreasing in $\mu \in (0, +\infty)$. Substituting $x_k = g_k(\mu)$ into Equation~\eqref{eq:de-MC1}, we have
\begin{equation}
\label{eq:de-mu}
\mu = \sum_{\ell \in \cK} d_\ell g_\ell(\mu).
\end{equation}
Clearly, the LHS of Equation~\eqref{eq:de-mu} is strictly increasing in $\mu$ and the RHS is positive and decreasing in $\mu$. Thus, Equation~\eqref{eq:de-mu} admits a unique solution, denoted by $\mu^*$. Let $x^*_k = g_k(\mu^*) > 0$ for each $k \in \cK$ and $\lambda^* = C'(\sum_{\ell \in \cK} d_\ell x^*_\ell)$.

Let $\bm{x}^* = \big( (x_i^t)_{t \in \cT_i} \big)_{i \in \cN}$ be a semi-symmetric strategy profile so that $x_i^t = x^*_k > 0$ for each battle $t$ with size $k$. The above analysis shows that this interior strategy profile $\bm{x}^*$ satisfies the first-order conditions with equalities. Since each payoff function $\Pi_i(\bm{x}_i, \bm{x}_{-i})$ is concave in $\bm{x}_i$, the semi-symmetric strategy profile $\bm{x}^*$ is a Nash equilibrium.

For the uniqueness, we directly follow Theorem~2(ii) in \cite{xu2022networks}.
\end{proof}

\medskip

\begin{proof}[Proof of Proposition~\ref{prop:ue-eqm}]
Let $\bigl( \cN, \cT, \bm{\Gamma}, f(\cdot), (v_k)_{k\in\cK}, C(\cdot) \bigr)$ be a semi-symmetric conflict network. We shall construct an interior symmetric Nash equilibrium under UE.

It is easy to see that each payoff function $\Pi_i(x_i, x_{-i})$ is concave in $x_i$. Thus, in equilibrium, each $x_i$ should satisfy the first-order conditions---each player $i$'s marginal benefit from exerting increment effort in each battle $t$ must be no more than the marginal cost, with equality if the solution $x_i$ is interior. That is,
$$ \sum_{t \in \cT_i} v^t \cdot \frac{f'(x_i) \cdot \Big[ \sum\limits_{j \in \cN^t, j \ne i} f(x_j) \Big]}{\Big[ \sum\limits_{j \in \cN^t} f(x_j) \Big]^2}  \le  C' \Big( \sum_{\ell \in \cK} d_\ell x_i \Big) \cdot \Big( \sum_{\ell \in \cK} d_\ell \Big) \text{ with equality if $x_i > 0$}. $$

We shall construct a symmetric strategy profile $\bm{x} = (x, x, \ldots, x)$ satisfying the above first-order condition with equality. That is, we try to provide a solution $x$ for the system of equations
\begin{align}
\sum_{\ell \in \cK} d_\ell v_\ell \cdot \frac{\ell-1}{\ell^2} \cdot \frac{f'(x)}{f(x)}	& =  C'(\mu) \cdot \Big( \sum_{\ell \in \cK} d_\ell \Big),	\label{eq:ue-FOC1}	\\
\mu		& = \sum_{\ell \in \cK} d_\ell x.	\label{eq:ue-MC1}
\end{align}
For each $\mu > 0$, we have $C'(\mu) > 0$. By Lemma~\ref{lem:h}, $\frac{f'(x)}{f(x)}$ is strictly decreasing in $x \in (0, +\infty)$, $\lim\limits_{x \to 0+} \frac{f'(x)}{f(x)} = +\infty$, and $\lim\limits_{x \to +\infty} \frac{f'(x)}{f(x)} = 0$. Thus, there exists a unique solution for Equation~\eqref{eq:ue-FOC1}, denoted by $x = g(\mu)$. Clearly, $x = g(\mu) > 0$ for all $\mu > 0$. Since $C(\cdot)$ is strictly increasing, $g(\mu)$ is decreasing in $\mu \in (0, +\infty)$. Substituting $x = g(\mu)$ into Equation~\eqref{eq:ue-MC1}, we have
\begin{equation}
\label{eq:ue-mu}
\mu = \sum_{\ell \in \cK} d_\ell g(\mu).
\end{equation}
Clearly, the LHS of Equation~\eqref{eq:ue-mu} is strictly increasing in $\mu$ and the RHS is positive and decreasing in $\mu$. Thus, Equation~\eqref{eq:ue-mu} admits a unique solution, denoted by $\mu^u$. Let $x^u = g(\mu^u) > 0$ and $\lambda^u = C'(\sum_{\ell\in\cK} d_\ell x^u) \cdot (\sum_{\ell\in\cK} d_\ell)$.

Let $\bm{x}^u = (x^u, x^u, \ldots, x^u)$ be a symmetric strategy profile. The above analysis shows that the interior strategy profile $\bm{x}^u$ satisfies the first-order conditions with equality. Since each payoff function $\Pi^u_i(\bm{x}_i, \bm{x}_{-i})$ is concave in $\bm{x}_i$, the symmetric strategy profile $\bm{x}^u$ is a Nash equilibrium.

For the uniqueness, we directly follow Proposition~5 in \cite{xu2022networks}.
\end{proof}

\subsection{Proof of Theorem~\ref{thm:comparison-h}}

\begin{proof}[Proof of Theorem~\ref{thm:comparison-h}]
Suppose DE has a higher total effort than UE when $h$ is convex. Then we have $\sum_{\ell \in \cK} d_\ell x^*_\ell > \sum_{\ell \in \cK} d_\ell x^u$. Thus, by Equations~\eqref{eq:de-MC} and \eqref{eq:ue-MC}, we have
$$ \lambda^* = C' \Bigl( \sum_{\ell \in \cK} d_\ell x^*_\ell \Bigr) > C' \Bigl( \sum_{\ell \in \cK} d_\ell x^u \Bigr) = \frac{\lambda^u}{\sum_{\ell \in \cK} d_\ell}. $$

From Equation~\eqref{eq:de-FOC}, we have
$$ v_k \cdot \frac{k-1}{k^2} \cdot \frac{1}{\lambda^*} = h(x^*_k) \text{ for each $k \in \cK$}. $$
Summing across all $k \in \cK$ with weights $d_k$, we then have
$$ \sum_{k \in \cK} d_k v_k \cdot \frac{k-1}{k^2} \cdot \frac{1}{\lambda^*} = \sum_{k \in \cK} d_k h(x^*_k) \text{ or } \frac{\sum_{k \in \cK} d_k v_k \frac{k-1}{k^2} \frac{1}{\lambda^*}}{\sum_{\ell \in \cK} d_\ell} = \frac{\sum_{k \in \cK} d_k h(x^*_k)}{\sum_{\ell \in \cK} d_\ell}. $$

From Equation~\eqref{eq:ue-FOC}, we have
$$ \sum_{\ell \in \cK} d_\ell v_\ell \cdot \frac{\ell-1}{\ell^2} \cdot \frac{1}{\lambda^u} = h(x^u). $$

Since $h$ is convex, we have
$$ \frac{\sum_{k \in \cK} d_k h(x^*_k)}{\sum_{\ell \in \cK} d_\ell} \ge h\left(\frac{\sum_{k\in\cK} d_k x^*_k}{\sum_{\ell\in\cK} d_\ell}\right). $$
Thus, we have
\begin{align*}
\frac{\sum_{k \in \cK} d_k v_k \frac{k-1}{k^2} \frac{1}{\lambda^*}}{\sum_{\ell \in \cK} d_\ell}	& = \frac{\sum_{k \in \cK} d_k h(x^*_k)}{\sum_{\ell \in \cK} d_\ell} \ge h\left(\frac{\sum_{k\in\cK} d_k x^*_k}{\sum_{\ell\in\cK} d_\ell}\right)	\\
	& > h\left(\frac{\sum_{k\in\cK} d_k x^u}{\sum_{\ell\in\cK} d_\ell}\right) = h(x^u) = \sum_{\ell \in \cK} d_\ell v_\ell \frac{\ell-1}{\ell^2} \frac{1}{\lambda^u}.
\end{align*}
That is, $\frac{\lambda^u}{\sum_{\ell\in\cK} d_\ell} > \lambda^*$, which leads to a contradiction. Therefore, DE has a lower total effort than UE.

By the similar arguments, one can prove that DE has a higher total effort than UE if $h$ is concave, and DE has the same total effort as UE if $h$ is linear.
\end{proof}

\subsection{Proof of Theorem~\ref{thm:neutrality-tullock}}

For each $k \in \cK$, let
\begin{equation}
\label{eq:h}
h_k(x) = \frac{f_k(x)}{f'_k(x)}.
\end{equation}
Lemma~\ref{lem:h} implies that (1) $h_k(x)$ is strictly increasing in $x \in (0, +\infty)$; (2) $\lim\limits_{x \to 0+} h_k(x) = 0$; (3) $\lim\limits_{x \to +\infty} h_k(x) = +\infty$.

\begin{lem}
\label{lem:linear}
For two distinct $m$ and $n$ in $\cK$, if
\begin{equation}
\label{eq:h_k}
(d_m + d_n) \cdot h_m(\tfrac{d_m z_m + d_n z_n}{d_m+d_n})  =  d_m \cdot h_m(z_m) + d_n \cdot h_m(z_n)
\end{equation}
holds for any positive $z_m$ and $z_n$, then $h_m(z) = \frac{1}{r_m} z$ for some $r_m > 0$.
\end{lem}

\begin{proof}[Proof of Lemma~\ref{lem:linear}]
By Lemma~\ref{lem:h}, $\lim_{z \downarrow 0} h_n(z) = 0$. Letting $z_n \downarrow 0$ in Equation~\eqref{eq:h_k}, we have for each $z_m > 0$,
$$ (d_m + d_n) \cdot h_m(\tfrac{d_m z_m}{d_m + d_n})  =  d_m \cdot h_m(z_m). $$
Similarly, letting $z_m \downarrow 0$ in Equation~\eqref{eq:h_k}, we have for each $z_n > 0$,
$$ (d_m + d_n) \cdot h_m(\tfrac{d_n z_n}{d_m + d_n})  =  d_n \cdot h_m(z_n). $$
Then, for all $z_m > 0$ and $z_n > 0$, we have
\begin{align*}
(d_m + d_n) \cdot h_m(\tfrac{d_m z_m + d_n z_n}{d_m+d_n})	& = d_m \cdot h_m(z_m) + d_n \cdot h_m(z_n)	\\
													& = (d_m + d_n) \cdot h_m(\tfrac{d_m z_m}{d_m + d_n}) + (d_m + d_n) \cdot h_m(\tfrac{d_n z_n}{d_m + d_n}),
\end{align*}
that is,
$$ h_m(\tfrac{d_m z_m + d_n z_n}{d_m+d_n})  =  h_m(\tfrac{d_m z_m}{d_m + d_n}) + h_m(\tfrac{d_n z_n}{d_m + d_n}). $$
So for any $y > 0$ and $y' > 0$,
$$ h_m(y+y')  =  h_m(y) + h_m(y'). $$
Thus, $h_m$ is a linear function, i.e., $h_m(z) = \frac{1}{r_m} z$ for some $r_m > 0$.
\end{proof}

\begin{proof}[Proof of Theorem~\ref{thm:neutrality-tullock}]
For any semi-symmetric conflict network $\bigl( \cN, \cT, \bm{\Gamma}, (f_k, v_k)_{k\in\cK}, C \bigr)$, one can have analogous equilibrium characterizations as in Propositions~\ref{prop:de-eqm} and \ref{prop:ue-eqm}:
\begin{align*}
v_k \cdot \frac{k-1}{k^2} \cdot \frac{f_k'(x^*_k)}{f_k(x^*_k)}	& = C' \Big( \sum_{\ell\in\cK} d_\ell x^*_\ell \Big) \text{ for each $k \in \cK$,}	\\
\sum_{\ell\in\cK} d_\ell v_\ell \cdot \frac{\ell-1}{\ell^2} \cdot \frac{f_\ell'(x^u)}{f_\ell(x^u)}	& = C' \Big( \sum_{\ell\in\cK} d_\ell x^u \Big) \cdot \Big( \sum_{\ell\in\cK} d_\ell \Big).
\end{align*}
The second equation can be rewritten as
$$ \sum_{\ell\in\cK} \frac{d_\ell v_\ell \cdot \frac{\ell-1}{\ell^2} \cdot \frac{f_\ell'(x^u)}{f_\ell(x^u)}}{C'(\sum_{\ell'\in\cK} d_{\ell'} x^u)} = \sum_{\ell\in\cK} d_\ell. $$

By the assumption, the total efforts of each player are the same, i.e., $\sum_{\ell\in\cK} d_\ell x^*_\ell = \sum_{\ell\in\cK} d_\ell x^u$. Then we have
\begin{align}
\sum_{\ell\in\cK} d_\ell    & =  \sum_{\ell\in\cK} \frac{d_\ell v_\ell \cdot \frac{\ell-1}{\ell^2} \cdot \frac{f_\ell'(x^u)}{f_\ell(x^u)}}{C'(\sum_{\ell'\in\cK} d_{\ell'} x^u)}  = \sum_{\ell\in\cK} \frac{d_\ell v_\ell \cdot \frac{\ell-1}{\ell^2} \cdot \frac{f_\ell'(x^u)}{f_\ell(x^u)}}{C'(\sum_{\ell'\in\cK} d_{\ell'} x^*_{\ell'})}   \nonumber\\
    & = \sum_{\ell\in\cK} \frac{d_\ell v_\ell \cdot \frac{\ell-1}{\ell^2} \cdot \frac{f_\ell'(x^u)}{f_\ell(x^u)}}{v_\ell \cdot \frac{\ell-1}{\ell^2} \cdot \frac{f_\ell'(x^*_\ell)}{f_\ell(x^*_\ell)}} = \sum_{\ell\in\cK} \frac{d_\ell h_\ell(x^*_\ell)}{h_\ell(x^u)}.    \label{eq:key}
\end{align}

\bigskip

Pick any two distinct indexes $m$ and $n$ in $\cK$, and a positive constant $z$. Let $z_m = z + \frac{\varepsilon}{d_m}$, $z_n = z - \frac{\varepsilon}{d_n}$, and $z_k = z$ for each $k \in \cK$ with $k \ne m, n$ (if any). Clearly, the vector $(z_k)_{k\in\cK}$ satisfies $\sum_{\ell\in\cK} d_\ell z_\ell = \sum_{\ell\in\cK} d_\ell z$. Based on the similar arguments in Remark \ref{rmd:de-eqm}, one can find a new semi-symmetric conflict network $\hat{H} = \bigl( \cN, \cT, \bm{\Gamma}, (f_k, \hat{v}_k)_{k\in\cK}, C \bigr)$ such that $\bm{z} = (\bm{z}_i)_{i\in\cN} = \bigl( (z_i^t)_{t\in\cT} \bigr)_{i\in\cN}$ is the unique Nash equilibrium therein, where $z_i^t = z_k$ for each battle $t$ with size $k$. Notice that $\hat{H}$ is also in the collection $H$. By the assumption, the symmetric Nash equilibrium $\bm{z}^u = (z^u, z^u, \ldots, z^u)$ under UE has the same total effort for each player with $\bm{z}$. That is, $\sum_{\ell\in\cK} d_\ell z^u = \sum_{\ell\in\cK} d_\ell z_\ell = \sum_{\ell\in\cK} d_\ell z$. Thus, $z^u = z$. That is, $(z, z, \ldots, z)$ is a symmetric Nash equilibrium of $\hat{H}$ under UE.

Till here, we have that $\bm{z} = (\bm{z}_i)_{i\in\cN}$ is the semi-symmetric Nash equilibrium of $\hat{H}$ under DE, and $\bm{z}^u = (z, z, \ldots, z)$ is a symmetric Nash equilibrium of $\hat{H}$ under UE. By repeating the similar arguments in Equation~\eqref{eq:key}, we obtain
\begin{align*}
\frac{d_m h_m(z)}{h_m(z)} + \frac{d_n h_n(z)}{h_n(z)} + \sum_{k \ne m, n} d_k   & = \sum_{\ell\in\cK} d_\ell = \sum_{\ell\in\cK} \frac{d_\ell h_\ell(z_\ell)}{h_\ell(z)} \\
    & = \frac{d_m h_m(z_m)}{h_m(z)} + \frac{d_n h_n(z_n)}{h_n(z)} + \sum_{k \ne m, n} \frac{d_k h_k(z_k)}{h_k(z)}    \\
    & = \frac{d_m h_m(z+\frac{\varepsilon}{d_m})}{h_m(z)} + \frac{d_n h_n(z-\frac{\varepsilon}{d_n})}{h_n(z)} + \sum_{k \ne m, n} \frac{d_k h_k(z)}{h_k(z)},
\end{align*}
and hence
\begin{equation}
\label{eq:m&n}
\frac{d_m h_m(z+\frac{\varepsilon}{d_m})}{h_m(z)} + \frac{d_n h_n(z-\frac{\varepsilon}{d_n})}{h_n(z)} = \frac{d_m h_m(z)}{h_m(z)} + \frac{d_n h_n(z)}{h_n(z)}.
\end{equation}
Therefore,
$$ \frac{1}{h_m(z)} \cdot \frac{h_m(z+\frac{\varepsilon}{d_m}) - h_m(z)}{\frac{\varepsilon}{d_m}} = \frac{1}{h_n(z)} \cdot \frac{h_n(z) - h_n(z-\frac{\varepsilon}{d_n})}{\frac{\varepsilon}{d_n}}. $$

Letting $\varepsilon \to 0$, we have that for each $z > 0$,
$$ \frac{h_m'(z)}{h_m(z)} = \frac{h_n'(z)}{h_n(z)} > 0. $$
Then $\frac{\dif}{\dif z} \left[ \log h_m(z) \right] = \frac{\dif}{\dif z} \left[ \log h_n(z) \right]$, and hence $h_m(z) = c \cdot h_n(z)$ for each $z > 0$, where $c$ is a positive constant number.

Substituting into Equation~\eqref{eq:m&n}, we have
$$ d_m + d_n = \frac{d_m h_m(z_m)}{h_m(z)} + \frac{d_n h_n(z_n)}{h_n(z)} = \frac{d_m h_m(z_m)}{h_m(z)} + \frac{d_n h_m(z_n)}{h_m(z)}, $$
or
$$ (d_m + d_n) h_m(z)  =  d_m h_m(z_m) + d_n h_m(z_n), $$
where $(d_m+d_n) z = d_m z_m + d_n z_n$.
Since $z$ and $\varepsilon$ are flexible, the equation above holds for all $z_m > 0$ and $z_n > 0$. By Lemma~\ref{lem:linear}, $h_m(x) = \frac{1}{r_m} x$ for some $r_m > 0$. Therefore,
$$ \bigl( \log f_m(x) \bigr)' = \frac{f'_m(x)}{f_m(x)} = \frac{1}{h_m(x)} = \frac{r_m}{x} = (r_m \log x)'. $$
Since $f_m(0) = 0$, we have $f_m(x) = A_m x^{r_m}$.

Since the index $m$ is randomly picked, each $f_k(x)$ should be of the power form $A_k x^{r_k}$ for some $A_k$ and $r_k$.

Lastly, since we have assumed that each $f_k(x)$ satisfies $f_k'(x) > 0$ and $f_k''(x) \le 0$ for all $x > 0$, each parameter $r_k$ should lie in $(0, 1]$.
\end{proof}

\end{document}